\documentclass{article}

\usepackage{fullpage}
\usepackage[english]{babel}
\usepackage{verbatim}

\usepackage{amsmath}
\usepackage{amssymb}
\usepackage{amsthm}
\usepackage{algorithm}
\usepackage{algorithmic}
\usepackage{graphicx}
\newtheorem{theorem}{Theorem}
\newtheorem{lemma}{Lemma}

\title{On the computational complexity of solving stochastic
  mean-payoff games\thanks{Work supported by {\em Center for
      Algorithmic Game Theory}, funded by the Carlsberg
    Foundation.}}
\author{Vladimir Gurvich \and Peter Bro Miltersen}
\begin{document}

\maketitle
\begin{abstract}
We consider some well known families of two-player, zero-sum,
 turn-based,
perfect information games that
can be viewed as specical cases of Shapley's stochastic
games.
We show that the following tasks are polynomial time equivalent:
\begin{itemize}
\item{}Solving simple stochastic games, 
\item{}solving stochastic mean-payoff games with rewards and
  probabilities given in unary, and
\item{}solving stochastic mean-payoff games with rewards and
  probabilities given in binary.
\end{itemize}
\end{abstract}

\section{Introduction}
We consider some well known families of two-player, zero-sum,
 turn-based,
perfect information games that
can be viewed as specical cases of Shapley's stochastic
games \cite{Shapley}. They have appeared under various names 
in the literature in the
last 50 years and variants of them have been rediscovered many times
by various research communities. For brevity, in this paper we 
shall refer to them by the name of the researcher who first 
(as far as we know) singled them out.
\begin{itemize}
\item{}{\bf Condon games} \cite{Condon92} (a.k.a. simple stochastic
  games). A Condon game is given by a directed graph $G = (V,E)$ with
a partition of the vertices into $V_1$ (vertices beloning to Player
  1, $V_2$ (vertices belonging to Player 2), $V_R$ (random
  vertices), and a special terminal vertex {\bf 1}.
Vertices of $V_R$ have exactly two outgong arcs, the terminal
vertex 
{\bf 1}
has none,
 while all vertices in $V_1, V_2$ have at least
one outgoing arc.
Between moves, a pebble is resting at one of the vertices $u$.
If $u$ belongs to a player, this player should strategically
pick an outgoing arc from $u$ and move the pebble along this edge to
another vertex. If $u$ is a vertex in $V_R$, nature picks
an outgoing arc from $u$ uniformly at random and moves the pebble
along this arc. The objective of the game for Player 1 is to
reach {\bf 1} and should play so as to maximize his probability of
doing so. The objective for Player 2 is to prevent Player 1 from
reaching {\bf 1}.
\item{}{\bf Gillette games} \cite{Gil}. A Gillette game $G$ 
is given by a finite set of states $S$, partioned into $S_1$ 
(states belonging to Player 1) and $S_2$ (states belonging to
Player 2). To each state $u$ is associated a finite set of possible actions.
To each such action is associated a real-valued {\em
  reward} and a probability distribution on states. 
At any point in time of play, the game is in a particular
state $i$. The player to move chooses an action strategically
and the corresponding award is paid by Player 2 to Player 1.
Then, nature chooses the next state at random according to the probability
distribution associated with the action. The play continues forever 
and the accumulated
reward may therefore be unbounded. Fortunately, there are ways of associating
a finite payoff to the players in spite of this and more ways than
one (so $G$ is not just
one game, but really a family of games): For {\em discounted}
Gillette games, we fix a {\em discount factor} $\beta \in (0,1)$
and define the payoff to Player 1 to be 
\[ \sum_{i=0}^{\infty} \beta^i r_i \]  where $r_i$ is the reward 
incurred at stage $i$ of the game. We shall denote the resulting
game $G_\beta$. For {\em undiscounted} Gillette
game we define the payoff to Player 1 to be the {\em limiting average} payoff
\[ \liminf_{n \rightarrow \infty} (\sum_{i=0}^n
r_i)/(n+1). \]
We shall denote the resulting game $G_1$.
\end{itemize}
Undiscounted Gillette games have recently been referred to as
{\em stochastic mean-payoff} games in the computer science literature
\cite{Chatterjee}. A natural restriction of Gillette games is
to {\em deterministic} transitions (i.e., all probability distributions
put all probability mass on one state). This class of games
has been studied in the computer science literature under the names
of cyclic games \cite{GKK} and mean-payoff games \cite{ZP}.

A {\em strategy} for a game is a (possibly randomized) procedure
for selecting which arc or action to take, given the history of the
play so far. A {\em pure, positional strategy} is the very special case 
of this where the choice is
deterministic and only depends on the current vertex (or state), i.e.,
a pure, positional strategy is simply a map from vertices (for
Gillette games, states) to
vertices (for Gillette games, actions).

A strategy $x^*$ for Player 1 is said to be {\em optimal} if for all
vertices (states) $i$ it holds that,
\begin{equation}\label{opt1}
\inf_{y \in S_2} \mu^i(x^*,y) \geq \sup_{x \in S_1} \inf_{y \in S_2}
\mu^i(x, y)
\end{equation} 
where $S_1$ $(S_2)$ is the set of strategies for Player 1 (Player 2)
and $\mu^i(x, y)$ is the probability that Player 1 will end up in
{\bf 1} (for the case of Condon games) or the expected payoff of
Player 1 (for the case of Gillette games) when players play using
the strategy profile $(x,y)$ and the play starts in vertex (state) $i$.
Similarly, a strategy $y^*$ for Player 2 is said to be optimal if
\begin{equation}\label{opt2}
 \sup_{x \in S_1} \mu^i(x,y^*) \leq \inf_{y \in S_2} \sup_{x \in S_1}
\mu^i(x, y).
\end{equation}
For all games described here, a proof of Liggett and Lippman 
\cite{LL} (fixing a bug of a proof of Gillette \cite{Gil}) shows
that there are optimal, pure, positional
strategies and that a pair of such strategies form 
an exact Nash equilibrium of
the game. These facts imply that when testing whether conditions
(\ref{opt1}) and (\ref{opt2}) holds, it is enough to take the
inifima and suprema over the finite set of pure, positional strategies of the players.

In this paper, we consider {\em solving} games. By solving a game
we mean the
task of computing a pair of optimal pure, positional strategies, 
given a description of the game as input\footnote{One may also 
define solving a game as computing its value (or comparing its value
to a fixed number, as in \cite{Condon92}). For the games considerd
here, this is polynomial time
(Turing) equivalent to finding optimal strategies. Our reductions are
 more coveniently described in
terms of finding optimal strategies rather than values.}. To be
able to finitely represent the games, we assume that the discount
factor, rewards and probabilities are rational numbers and given
as fractions.

It is well known that Condon games can be seen as a special case of
undiscounted Gillette games (as described in the proof of Lemma 4
below), but a priori, solving Gillette games could be harder. 
A recent paper by Chatterjee and Henzinger \cite{Chatterjee} shows that solving
so-called stochastic parity games \cite{CJH1,CJH2} reduces to 
solving undiscounted
Gillette games. This motivates the study of the complexity of the
latter task. We show that the extra expressive power 
(compared to Condon games) 
of having rewards during the game in fact does not change the
computational complexity of solving the games. More precisely,
our main theorem is:
\begin{theorem}
The following tasks are polynomial time equivalent:
\begin{enumerate}
\item{}Solving Condon games (a.k.a.., simple stochastic games)
\item{}Solving undiscounted Gillette games (a.k.a, stochastic mean-payoff games) with rewards and probabilities represented in binary notation.
\item{}Solving undiscounted Gillette games with rewards and
  probabilities represented in unary notation.
\item{}Solving discounted Gillette games with discount factor, rewards
  and probabilities represented in binary notation.
\end{enumerate}
\end{theorem}
In particular, there is a pseudopolynomial time algorithm for solving
undiscounted Gillette games if and only if there is a polynomial
time algorithm for this task.
The theorem follows from the Lemmas 2,3,4 below and the fact
that solving games with numbers in the input represented in unary trivially
reduces to solving games with numbers in the input represented in binary.
The proof techniques are fairly standard (although coming from
two different communities), but we find it worth pointing
out that they together imply the theorem above since it is relevant,
did not seem to be known\footnote{Although Condon \cite{Condon92}
  observed
that the case of Gillette games with {\em immediate} rewards reduces
to Condon games  and Zwick and Paterson \cite{ZP} that {\em deterministic}
Gillette games reduce to Condon games.}, and may even be considered slightly surprising, as
deterministic undiscounted Gillette games can be solved in 
pseudopolynomial time \cite{GKK,ZP}, while
solving them in polynomial time remains a challenging open problem.
An even more challenging problem is solving simple stochastic games in
polynomial time, so our theorem may be interpreted as a hardness
result. Note that a ``missing bullet'' in the theorem is 
solving discounted Gillette games given in unary notation. It is in fact
known that this can be done in polynomial time (even if only the
discount factor is given in unary while rewards and probabilities are
given in binary), see Littman \cite[Theorem 3.4]{litt}.
\section{Proofs}

\begin{lemma}\label{limit}
Let $G$ be a Gillette game with $n$ states and all transition probabilities and
rewards being fractions with integral numerators and denominators, all of
absolute value at most $M$. 
Let $\beta^* = 1 - ((n!)^2 2^{2n+3} M^{2 n^2})^{-1}$ and let 
$\beta \in [\beta^*, 1]$. Then, any optimal pure stationary 
strategy (for either
player) in the discounted Gillette game $G_{\beta}$ is also an optimal 
strategy in the undiscounted Gillette game $G_1$.
\end{lemma}
\begin{proof}
The fact that {\em some} $\beta^*$ with the 
desired property exists is explicit in the proof 
of Theorem 1 of Liggett and Lippman
\cite{LL}. 
Here, we derive a concrete value for $\beta^*$. 
From the proof of Liggett and Lippman, we have that for 
$x^*$ to be an optimal pure stationary strategy (for Player 1) 
in $G_1$, it is
sufficient to be an optimal pure stationary strategy in $G_\beta$ for
all values of $\beta$ sufficiently close to $1$, i.e., to satisfy 
the inequalities
\[ \min_{y \in S'_2} \mu^i_\beta(x^*,y) \geq \max_{x \in S'_1} \min_{y \in S'_2}
\mu^i_\beta(x, y) \] for all states $i$ and for all values of 
$\beta$ sufficiently close to $1$, where $S'_1$ ($S'_2$) is the set
of pure, positional, strategies for Player 1 (2) and $\mu^i_\beta$ is
the expected payoff when game starts in position $i$ and the
discount
factor is $\beta$.
Similarly, for 
$y^*$ to be an optimal pure stationary strategy (for Player 1) 
in $G_2$, it is
sufficient to be an optimal pure stationary strategy in $G_\beta$ for
all values of $\beta$ sufficiently close to $1$, i.e., to satisfy 
the inequalities
\[ \max_{x \in S'_1} \mu^i_\beta(x,y^*) \leq \min_{y \in S'_2} \max_{x \in S'_1}
\mu^i_\beta(x, y). \] 
So, we can prove the lemma by showing that 
for all states $i$ and {\em all} pure stationary 
strategies $x,y,z,u$, the sign of 
$\mu^i_\beta(x,y) - \mu^i_\beta(z,u)$ 
is the same for all $\beta \geq \beta^*$. For fixed
strategies $x,y$ we have that $v_i = \mu^i_\beta(x,y)$ is the expected
total reward in a {\em discounted Markov process} and is therefore
given by the formula (see \cite{Howard}) 
\begin{equation}\label{closed}
v = (I - \beta Q)^{-1} r,
\end{equation}
where $v$ is the vector of
$\mu_\beta(x,y)$ values, one for each state, $Q$ is the matrix of
transition probabilities and $r$ is the vector of rewards (note that
for {\em fixed} positional strategies $x,y$, 
rewards can be assigned to states in 
the natural way). 
Let $\gamma = 1 - \beta$. Then, (\ref{closed}) 
is a system of linear equations in 
the unknowns $v$, where each coefficient is of the form $a_{ij} \gamma
+ b_{ij}$ where $a_{ij}, b_{ij}$ are rational numbers with numerators
with absolute value bounded by $2M$ and with denominators with
absolute value bounded by $M$. 
By multiplying the equations with all denominators, we can in fact
assume that $a_{ij}, b_{ij}$ are integers of absolute value less than $2M^n$.
Solving the equations using Cramer's rule, we may write
an entry of $v$ as a quotient between determinants of $n \times n$ 
matrices containing terms of the form $a_{ij} \gamma + b_{ij}$.
The determinant of such a
matrix is a polynomial in $\gamma$ of degree $n$ with the coefficient
of each term being of absolute value at most $n! (2 M^n)^n = n! 2^n M^{n^2}$. We
denote these two polynomials $p_1, p_2$. 
Arguing similarly about $\mu_\beta(z,u)$ and deriving corresponding
polynomials $p_3, p_4$, we have that 
$\mu^i_\beta(x,y) - \mu^i_\beta(z,u) \geq 0$ is equivalent to
$p_1(\gamma)/p_2(\gamma) - p_3(\gamma)/p_4(\gamma) \geq 0$, i.e.,
$p_1(\gamma)p_4(\gamma) - p_3(\gamma)p_2(\gamma) \geq 0$.
Letting $q(\gamma) = p_1(\gamma)p_4(\gamma) - p_3(\gamma)p_2(\gamma)$,
we have that $q$ is a polynomial in $\gamma$, with integer
coefficients, all of absolute value at most $R= 2(n!)^2 2^{2n} M^{2n^2}$.
Since $1 - \beta^* < 1/(2R)$, the sign of $q(\gamma)$ is the same
for all $\gamma \leq 1- \beta^*$, i.e., for all $\beta \geq \beta^*$.
This completes the proof.
\end{proof}
\begin{lemma}
Solving undiscounted Gillette games (with binary representation of
rewards and probabilities) polynomially reduces to solving discounted Gillette
games (with binary representation of discount factor, rewards, and probabilities).
\end{lemma}
\begin{proof}
This follows immediately from Lemma \ref{limit} by observing that
the binary representation of the number 
$\beta^* = 1 - ((n!)^2 2^{2n+3} M^{2 n^2})^{-1}$ has length polynomial 
in the size of the representation of the game.
\end{proof}
\begin{lemma}
Solving discounted Gillette game (with binary representation of
discount factor, rewards, and probabilities) polynomially reduces to solving Condon games.
\end{lemma}
 \begin{proof}
Zwick and Paterson \cite{ZP} considered solving {\em deterministic}
discounted Gillette games, i.e., Gillette
games where the action deterministically determines the transition
taken. 
It is
natural to try to generalize their reduction so that it also works for
general discounted Gillette games. 
Since Condon games allows for vertices making random choices, 
a natural attempt is to simply 
simulate a stochastic 
transition by such random vertices. Such a
generalization is made even easier by the fact that Zwick and
Paterson proved that solving ``augmented'' Condon games where
random vertices are allowed to take choices given by arbitrary discrete
distributions with rational probability weights (represented in binary)
is polynomially equivalent to solving ``plain'' Condon games. 
We find that the reduction outlined above is indeed correct, even though the
correctness proof of Zwick and Paterson has to be modified slightly
compared to their proof.
The details follow.

We are given as input a Gillette game form $G$ and a discount
factor $\beta$ and must produce an augmented Condon 
game $G'$ whose solution yields
the solution to the Gillette game $G_\beta$.
First, we affinely scale and translate all rewards of $G$ so 
that they are in the interval $[0,1]$. This does not influence
the optimal strategies. 
Vertices of $G'$ include all states of $G$ (belonging to the same player in
$G'$
as in $G$), and, in addition, a random
vertex $w_{u,A}$ for each possible action $A$ of 
each state $u$ of $G$.
We also add a ``trapping'' vertex {\bf 0} with a single arc to
itself. It does not matter which player it belongs to.
We construct the arcs of $G'$ by adding, for each (state,action) pair
$(u,A)$ the ``gadget'' indicated in Figure \ref{fig}.
\begin{figure}
\begin{center}
\includegraphics[height=6cm]{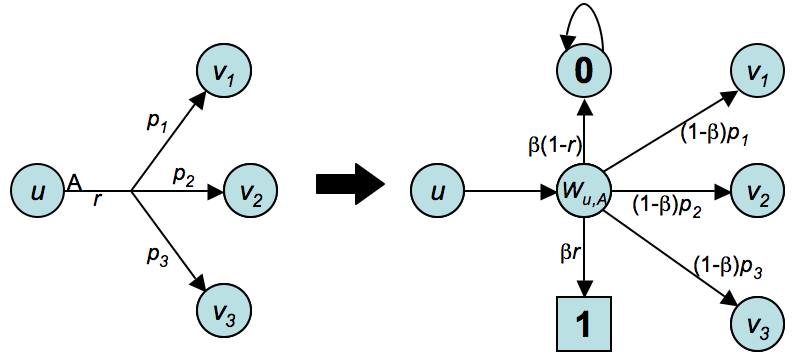}
\caption{\label{fig} Reducing discounted Gillette games to Condon games}
\end{center}
\end{figure}
To be precise, if the action has reward $r$ and leads to states
$v_1, v_2, \ldots, v_k$ with probability weights $p_1, p_2, \ldots, p_k$,
we include in $G'$ an arc from $u$ to $w_{u,A}$, arcs from
$w_{u,A}$ to $v_1, \ldots, v_k$ with probability weights
$(1-\beta)p_1,\ldots,(1-\beta)p_k$, an arc from $w_{u,A}$ to 
{\bf 0} with probability weight $\beta(1-r)$ and
finally an arc from $w_{u,A}$ to the terminal {\bf 1} with probability weight
$\beta r$. 

There is clearly a 1-1 correspondence between pure stationary strategies in $G$
and in $G'$. Thus, we are done if we show that the
optimal strategies coincide. To see this, fix a strategy profile 
for the two players and consider play starting in any vertex $u$. 
By construction, if the expected reward of the play in $G$ is $h$,
the probability that the play in $G'$ ends up in {\bf 1} is exactly $\beta h$.
Therefore, the two games are strategically equivalent.
\end{proof}
\begin{lemma}
Solving Condon games polynomially reduces to solving 
undiscounted Gillette games
with {\em unary} representation of rewards and probabilities.
\end{lemma}
\begin{proof}
We are given a Condon game $G$ (a ``plain'' one, using
the terminology of the previous proof) and must construct
an undiscounted Gillette game $G'$. States of $G'$ will coincide
with vertices of $G$, with the states of $G'$ including the special
terminals {\bf 1}. Vertices $u$ belonging to a player in 
$G$ belongs to the same player in $G'$. For each outgoing
arc of $u$, we add an action in $G'$ with reward 0, and with a
deterministic transition to the endpoint of the arc of $G$.
Random vertices of $G$ can
be assigned to either player in 
$G'$, but he will only be given
a single ``dummy choice'': If the random vertex has arcs to $v_1$
and $v_2$, we add a single action in $G'$ with reward
$0$ and transitions into $v_1$, $v_2$, both with probability weight
$1/2$. The terminal {\bf 1} can be assigned to either player
in $G'$, but again he will be given only a dummy choice: We add
a single action with reward 1 from {\bf 1} 
and with a transition back into {\bf 1} 
with probability weight $1$.

There is clearly a 1-1 correspondence between pure stationary strategies in $G$
and strategies in $G'$. Thus, we are done if we show that the
optimal strategies coincide. To see this, fix a strategy profile 
for the two players and consider play starting in any vertex $u$. 
By construction, if the probability of the play ending up in {\bf 1}
in $G$ is $q$, the expected limiting average reward of the play in 
$G'$ is also $q$.
Therefore, the two games are strategically equivalent, and we are done.
\end {proof}
\section{Open problems}
Undiscounted Gillette games can be seen as generalizations of Condon games and
yet they are computationally equivalent. It is interesting to ask if 
further generalizations of Gillette games are also equivalent to solving Condon
games. It seems natural to restrict attention to cases where
it is known that optimal, positional strategies exists. This
precludes general stochastic games (but see \cite{HE}). An interesting
class of games generalizing undiscounted Gillette games was considered
by Filar \cite{Filar}. Filar's games allow simultaneous moves by
the two players. However, for any position, the probability
distribution on the next position can depend on the action of 
one player only. Filar
shows that his  games
are guaranteed to have optimal, positional 
strategies. The optimal strategies are not necessarily pure, but
the probabilities they assign to actions are guaranteed to be rational
numbers if rewards and probabilities are rational numbers. So, we ask: Is
solving Filar games polynomial time equivalent to solving Condon games?
\bibliographystyle{plain}

\end{document}